\documentclass[journal,onecolumn,12pt]{IEEEtran}
\IEEEoverridecommandlockouts
\usepackage{amsmath, amssymb, stmaryrd} 
\usepackage{graphicx}
\usepackage{amsmath}
\usepackage{color}
\usepackage{hyperref}
\usepackage{multirow}
\usepackage{fancyhdr}
\usepackage[normal,bf,up]{caption2}
\usepackage{subfigure}
\usepackage{amsthm}
\usepackage[section]{placeins}
\usepackage{comment}
\usepackage{amssymb}
\usepackage{setspace}
\usepackage{mathtools, nccmath}

\begin{document}
\parindent=5mm

\DeclarePairedDelimiter{\nint}\lfloor\rceil
\newcommand{\beq}{\begin{equation}}
\newcommand{\eeq}{\end{equation}}
\newcommand{\Hb}{\mathbf {H}}
\newcommand{\F}{\mathbf {F}}
\newcommand{\B}{\mathbf {B}}
\newcommand{\I}{\mathbf {I}}
\newcommand{\p}{\boldsymbol {p}}
\newcommand{\bv}{\boldsymbol {b}}
\newcommand{\matr}[1]{{\mathbf{#1}}}
\renewcommand{\vec}[1]{{\boldsymbol{#1}}}
\newcommand{\x}{\boldsymbol {x}}
\newcommand{\n}{\boldsymbol {n}}
\newcommand{\s}{\boldsymbol {s}}
\newcommand{\w}{\boldsymbol {w}}
\newcommand{\vv}{\boldsymbol {v}}
\newcommand{\cx}{\boldsymbol {c}}
\newcommand{\rx}{\boldsymbol {r}}
\newcommand{\y}{\boldsymbol {y}}
\newcommand{\z}{\boldsymbol {z}}
\newcommand{\R}{\mathbf {R}}
\newcommand{\A}{\mathbf {A}}
\newcommand{\mL}{\mathcal{L}}
\newcommand{\raw}{\rightarrow}

\newtheorem{lemma}{Lemma}
\newtheorem{theorem}{Theorem}
\newcommand{\gvect}[1]{{\boldsymbol{#1}}} 
\newcommand{\pound}{\operatornamewithlimits{\#}}
\newcommand{\stxt}[1]{\ensuremath{_{\text{#1}}}}
\renewcommand{\topfraction}{1}	
\renewcommand{\bottomfraction}{1}	
\renewcommand{\textfraction}{0}	
\newcommand{\defeq}{\stackrel{\textrm{def}}{=}}



\title { An Approach to Asynchronous Unsourced Random Access}

\author{
  Alireza Karami \texttt{akarami@dal.com}\\
  \and
  Dmitry Trukhachev \texttt{dmitry@dal.ca}
}

\maketitle

\begin{abstract}
 In this work we propose an approach to construct a fully asynchronous unsourced random access communications (URA) system. Besides the common URA features such as the absence of the user identification information in the packet and the decoding oriented to receiving the content of the messages, there is no limitation on the user transmission timing. The active users can transmit their packet any time on demand. The proposed system belongs to the class of preamble-payload URA formats. Contrary to the typical role of the preamble to serve as a temporary user identification, the preamble serves for the purpose of timing acquisition.  The proposed system also employs a very small pool of preambles to resolve collisions of packets transmitted simultaneously. We provide a performance analysis and simulation results for the proposed system and demonstrate that for a wide range of signal-to-noise ratios (SNR)s the supported numbers of active users are close to the numbers of users supported by a slotted URA system counterpart.
\end{abstract}

\section{Introduction}

Unsourced multiple access (URA)~\cite{Pol17} is a technique for grant-free communication with the main target application in massive connectivity in the next generation of communication networks. In URA the packets transmitted by multiple users have to deliver information to the common receiver without a need to reveal and detect the identity of the transmitting users. The identity of a sender of a packet can be included in the payload if required. It is often assumed that the pool of users who can access the channel is large, but only a small fraction of the entire potential user population transmits concurrently. 

Several examples of recently proposed URA systems are based on the compressed sensing method~\cite{OP17,FenglerC19}. There are also URA systems~\cite{KAFP20,AMF20} based on ALOHA which apply error correction codes such as low density parity check (LDPC) codes to overcome collisions of users within the slots. In~\cite{TBKN20} a spatial coupling idea related to multi-user communications studied in~\cite{KT19} and~\cite{TS19} is applied to construct a URA transmission method. 
The majority of the URA systems proposed in the literature are synchronous and slotted, meaning that the packets are only allowed to be transmitted at the beginning of common time frames (slots).

Asynchronous transmission of packets can naturally occur in a URA system when the active users are not listening to a beacon from the receiver, due to wake up energy saving cycles, independent on demand transmission schedules etc. Another form of asynchronicity is the controlled asynchronicity when the users are selecting to transmit their data packets with pseudo-random delays. The latter approach allows to facilitate the reception using interference cancellation algorithms operating on a coupled chain of packets. In~\cite{TBKN20} randomized delays, which are known to the receiver (they are encoded in a synchronously transmitted preamble), are considered and help the system to utilize the benefit of spatial coupling. Here we note that the paradigm of spatial graph coupling which originated in error-correction coding~\cite{FZ99} is helping to create efficient systems in other areas of communications, including URA systems as in~\cite{TBKN20}.
We start with an overview of asynchronous access in context of literature of non-orthogonal multiple access (NOMA).

\subsection{Asynchronous Random Non-Orthogonal Multiple Access}

Asynchronous transmission is an indispensable component of the future multi-user systems. Analysis and simulation given in~\cite{LWYYBK21} confirm that achievable rates of asynchronous NOMA are higher than that for synchronous NOMA. In~\cite{ZHJ19} analyzing a two user asynchronous NOMA system shows the potential gain due to packet asynchronicity. Using the benefit of sampling diversity and for a sufficiently large frame size asynchronous NOMA outperforms synchronous NOMA. However, the starting times of packets for these systems need to be known by the receiver.

In grant-free asynchronous transmission, detecting packets (finding packets starting times) is a prerequisite for the interference cancellation process to be executed. The conventional way of time acquisition is based on correlation between existing preambles (known signals, attached to the beginning of packets) and the received signal. Improved composite pseudo-noise correlation (ICPC)~\cite{GLHLZ13} is one such method that is used to detect pseudo-noise (PN) preambles. This correlation technique has been analyzed in~\cite{PKHH15}. The analysis includes computation of the probability density function (PDF) of the absolute ICPC metric and its square which can be utilized to estimate real packet and false packet detection probabilities. It is worth to mention that even in synchronous transmission packets may get missed or false packets may be detected. This may happen in sparse code multiple access (SCMA) scheme, in which the decoder works on all active and inactive users~\cite{zte16}. There exist some studies and proposals that aim to decrease false and missed packet probabilities. For instance, in~\cite{KLKWNH17} a new codeword design and receiver processing for SCMA has been proposed which leads to a significant reduction in false alarms and missed packets. 

A combined orthogonal and non-orthogonal preamble structure is introduced in~\cite{KKH21} in order to provide support for higher number of active users and channel estimation for collided users.~\cite{KKH21} also proposes an interference cancellation decoder which covers more bits of the overlapped (interfering) user packets to reduce the interference. The authors simply  consider a decoding window with twice the decoding window length of the conventional grant free NOMA systems with orthogonal frequency division multiplexing (OFDM) signaling.

We also note that several non-orthogonal multiple access (NOMA) approaches developed in the literature are capable to operate in an grant-free mode. A few of existing NOMA methods have some elements of an asynchronous system but a fully asynchronous NOMA system is still a topic of active research. In addition, only a few of the existing NOMA proposals have been utilized in real practical implementations.

One of such systems capable to operate in an asynchronous regime is the multi-user shared access (MUSA)~\cite{MUSA16,YYYL17}. In MUSA the encoded bits of each user are first mapped to a standard symbol constellation such as quadrature phase shift keying (QPSK) and then spread using a complex spreading code of length $M$. The spreading codes need to have low cross-correlation in order for the multi-user decoding to be successful. 

The MUSA can be used in conjunction with a 2-step random channel access (2-step RACH) in the following way. The typical transmission is slotted (synchronous). A pool of spreading codes is established. Each transmitting user chooses one of the spreading codes from the pool randomly and transmits it as a preamble. The transmitted preamble is immediately followed by the payload which is spread with the same spreading code. The receiver starts with detecting the preambles received in the preamble slot and then preforms multi-user decoding of the payloads received in the payload slot. For any successfully detected and decoded packet the receiver returns an  acknowledgement (ACK) signal to the respective transmitter. In case of a collision where the same preamble has been transmitted by more than one user or in case when the preamble detection fails no messages are sent back to the related users. For the case when the preambles are detected but the payload decoding fails the failed users receive a negative acknowledgement (NACK).

While, potentially, there is a theoretic possibility that MUSA-style transmission may be used in an asynchronous fashion and also in conjunction with URA, a number of aspects related to timing acquisition, collision resolution and operation in a highly overloaded user regimes with non-orthogonal spreading codes have to be established. Instead of pursuing this avenue, we started a development of a distinct fully asynchronous URA system that steams from~\cite{TBKN20}, see Section~\ref{sec:FAT}. 
\subsection{Correlation Detection and Machine Learning}
Recently, machine learning (ML) algorithms have found their way into the physical layer of communication systems. ML algorithms can provide an improved signal detection~\cite{OKC17}, and more specifically preamble detection.
In~\cite{RMC20}, the advantage of ML algorithms such as neural networks (NN) and random forest in detecting starting times of packets are studied. A significant improvement over the correlation-based method has been demonstrated. Specially, by using a supervised learning NN method in an asynchronous transmission of short packets, 10 times improvement in reducing both the false and missed packets is reported. In~\cite{NVVD} a 1-dimensional convolutional NN has been used as an alternative to correlation based algorithms for detecting packets in IEEE 802.11. Although the authors claim their NN method is less complex and has a superior performance compared to the conventional methods, an appropriate comparison is not provided.

\subsection{Iterative Timing Acquisition and Multiple Access Decoding}
Packet acquisition is the first step at the  receiver which is required in order to start the decoding of the received packets. The challenge here is to detect as many of the transmitted packets as possible right from the beginning. The interference from other packets as well as noise cause some low power packets to be faded and unseen (missed packets), or some fake packets (false packets) to appear as real ones. These two undesired events worsen the overall performance of the receiver and can even make it fail. According to~\cite{CD94} the overall system capacity is significantly affected by success of the timing acquisition. 

In~\cite{T14} two different approaches for packet acquisition are studied and analyzed. The first approach includes sequential approach where timing acquisition is followed by the multiple access decoding (MUD). In the second approach packet acquisition is performed at the beginning of each decoding iteration of the MUD. According to~\cite{T14} in order to operate closer to the capacity,  iterative packet acquisition and MUD is inevitable. It is also demonstrated in~\cite{SRS08} that for high performance multi-user detection, information sharing between the timing acquisition and the receiver is essential. In~\cite{SRS08} at every iteration some new packets are detected and reported to the interleaved division multiple access (IDMA) detector. In addition,~\cite{RHC08} proposes a similar iterative approach for direct-sequence code-division multiple access (DS-CDMA) and introduces it as a high performance structure. \\

\subsection{Asynchronous URA}
To our knowledge, there is no work on asynchronous URA in the literature to date. In this paper we propose a fully asynchronous URA system in which there is no limitation for the users in terms of the transmission starting times.

Section~\ref{sec:FAT} outlines the work on the fully asynchronous URA case. We consider a system where each packet contains a known preamble and a payload, encoded for unsourced random access, and show our preliminary results. We introduce our receiver structure and an approach to further optimize the system in terms of performance.

\section{Fully Asynchronous URA Transmission} 
\label{sec:FAT}

Our purpose is to design a URA system which makes it possible to detect the starting times of the packets and minimize the probability of a packet collision at the same time. The system structure is as follows.

\subsection{System Model}
\label{sec:SM}
 Fig.~\ref{fig:RandTrans} shows how packets can arrive in a generic multi-user asynchronous system. Each packet includes two sections: preamble, used for timing acquisition, and payload, which carries the information bits.
 We divide the time into intervals of size $N_S$, which are called slots. However, the packets can start at any time. Whenever a packet is generated, it can be transmitted without any delay.

 \begin{figure}[t]
 	\centering
 	\setlength{\unitlength}{1mm}
 	\begin{picture}(140,50)
 	\put(0,0){\includegraphics[scale=0.5]{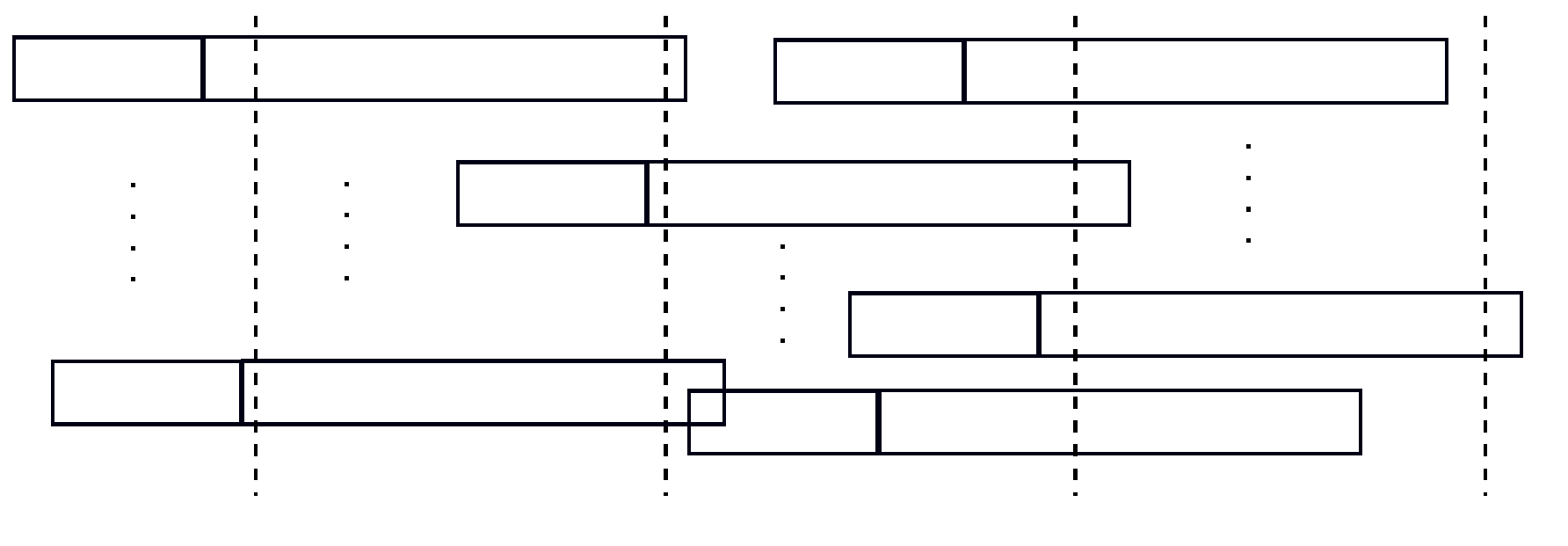}}
 	\put(77.5,1){Slot $j$}
 	\put(35,1){Slot $j-1$}
 	\put(114.5,1){Slot $j+1$}
 	
 	\put(4,45){\scriptsize{$\text{Pre}_{L,j-2}$}}
 	\put(8.5,14){\scriptsize{$\text{Pre}_{1,j-2}$}}
 	
 	\put(80,45){\scriptsize{$\text{Pre}_{L,j}$}}
 	\put(71.5,11){\scriptsize{$\text{Pre}_{1,j}$}}
 	
 	\put(46.5,33){\scriptsize{$\text{Pre}_{L,j-1}$}}
 	\put(87,20){\scriptsize{$\text{Pre}_{1,j}$}}
 	
 	\put(37,45){\scriptsize{$\text{Pay}_{L,j-2}$}}
 	\put(40,14){\scriptsize{$\text{Pay}_{1,j-2}$}}
 	
 	\put(110,45){\scriptsize{$\text{Pay}_{L,j}$}}
 	\put(102,11){\scriptsize{$\text{Pay}_{1,j}$}}
 	
 	\put(79,33){\scriptsize{$\text{Pay}_{L,j-1}$}}
 	\put(117,20){\scriptsize{$\text{Pay}_{1,j}$}}
 	
 	\put(-16,5){\scriptsize{Pay: Payload}}
 	\put(-16,0){\scriptsize{Pre: Preamble}}
 	
 	\end{picture}
 	\caption{Fully asynchronous transmission of packets that consist of preambles and payloads.}
 	\label{fig:RandTrans}
 \end{figure}
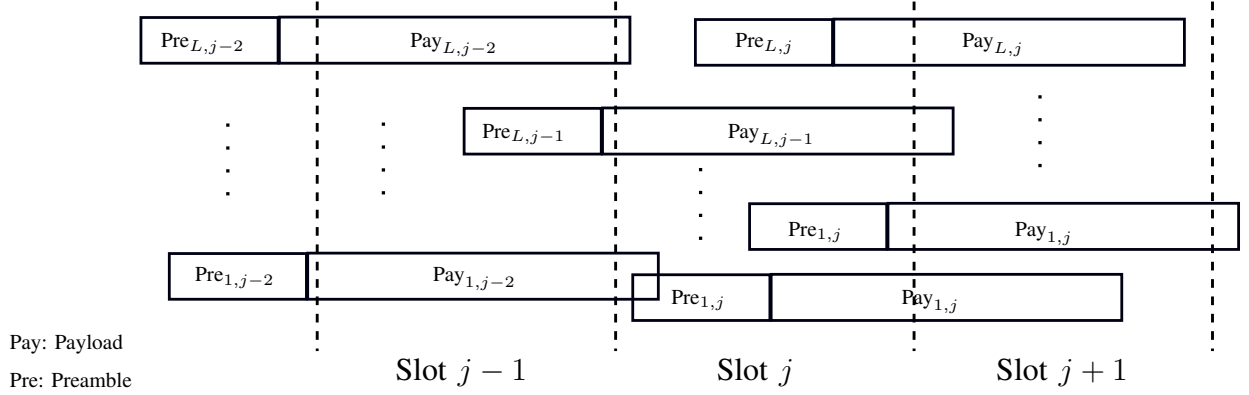
   
  We consider a packet format similar to the packet format in a URA with controlled delay transmission~\cite{TBKN20}. Suppose that $L=K_a$ active users are sharing a wireless channel to communicate to a single receiver. Information bits ($K$ bits) of each packet are first encoded for error-correction. This is followed by a binary phase shift keying (BPSK) modulation performed on the codeword bits. This results in a $N$ bit sequence, $\vec{v}_l^{(j)}$, for active user $l$, $l=1,2,\dots, K_a$, packet at time slot $j$:
\begin{align}
	\vec{v}_l^{(j)} = [v_{l,1}^{(j)}, v_{l,2}^{(j)}, \dots, v_{l,N}^{(j)}],
\end{align}
where $v_{l,n}^{(j)} \in \{-1,1\}$, $n=1, 2, \dots, N$ is the $n$th bit of the codeword $\vec{v}_l^{(j)}$. Here, we note again that the packet can start at any time within slot $j$. $N_p$ one, which later will be used to construct the preamble of the packet, are attached to the beginning of the codeword sequence: $ \hat{\vec{v}}_l^{(j)}=[\mathbf{1}, \vec{v}_l^{(j)}]= [\hat{v}_{l,1}^{(j)}, \hat{v}_{l,2}^{(j)}, \dots, \hat{v}_{l,N+N_p}^{(j)}]$, where $\mathbf{1}=[1, 1,\dots, 1]$ is of size $N_p$.
At the next stage of packet modulation the bits of the vector $\vec{v}_l^{(j)}$ are repeated $M$ times, resulting in the repeated-bits sequence
\begin{align}
\bar{\vec{v}}_l^{(j)} = [\bar{\vec{v}}_{l,1}^{(j)}, \bar{\vec{v}}_{l,2}^{(j)}, \dots, \bar{\vec{v}}_{l,N}^{(j)}],
\end{align}
where $\bar{\vec{v}}_{l,n}^{(j)}= [v_{n,1}^{(j)}, v_{n,1}^{(j)}, \dots, v_{n,1}^{(j)}]$ is an $1 \times M$ vector.

Then, a pair of a permutation and signature is chosen from a pool. The permutation is designed such that the first $MN_p$ positions, related to the preamble, do not change their value (i.e. are not permuted), but the remaining $MN$ permutation positions are (pseudo) randomly permuted. The signature is a randomly chosen sequence, where each element is chosen according to Bernoulli distribution from $\{-1/\sqrt{M}, 1/\sqrt{M}\}$. The values $-1/\sqrt{M}$ and $1/\sqrt{M}$ ensure that the energy of each data bit (information bits after possible FEC encoding) is 1. The permuted bit sequence, $\bar{\bar{\vec{v}}}_{l}^{(j)}$, is multiplied bit by bit with the signature sequence to provide $\tilde{\vec{v}}_{l}^{(j)}$, which is then transmitted over the channel. 

An event when two packets are encoded using the same (signature, permutation) pair and are transmitted by two different users at the same time is considered to be a collision. In case of no collisions it would be sufficient to use a single signature-permutation pair to modulate all packets and a use of a permutation-signature pool would not be required. The MUD at the receiver would be able to separate the packets as we have confirmed by the numerical results. The small preamble pool used in our system is introduced for the purpose of separating the simultaneously arriving packets. 



Therefore, packets of active users started at slot $j-1$ and $j$  accompanying with the channel's additive white Gaussian noise (AWGN) lead to a common received signal $\vec{y}^{(j)}=[y_1^{(j)}, y_2^{(j)}, \dots, y_{MN}^{(j)}]$ in slot $j$, where 
\begin{align}
y_t^{(j)} = \sum_{l=1}^{K_T} \left[\tilde{v}_{l,t+2MN-\tau_l^{(j-2)}+1}^{(j-2)}\text{I}(l^{(j-2)}) + \tilde{v}_{l,t+MN-\tau_l^{(j-1)}+1}^{(j-1)}\text{I}(l^{(j-1)}) + \tilde{v}_{l,t-\tau_l^{(j)}+1}^{(j)}\text{I}(l^{(j)})\right] + \n_t, \label{eq:ytj}
\end{align}
where $K_T$ is the total number of users. $\tau_l^{(j)}$ is the starting point of packet $l$ in slot $j$. $\text{I}(l^{(j)})$ is the indicator function that is 1 if packet $l$ has started in slot $j$, and 0 otherwise. $n_t$ is the independent and identically distributed (iid) zero mean AWGN sample with variance $\sigma^2$. Note that according to Fig.~\ref{fig:RandTrans}, packet that starts at slots $j-2$, $j-1$, $j$ can overlap with the $j$th slot.


\subsection{Arrival Times Detection} \label{sec:ATD}
Clearly, acquisition of starting times of the packets is necessary to start and perform the interference cancellation and error correction process at the receiver. There are two usual different approaches to detect starting times. One approach is based on correlation between a known preamble sequence and the received sequence (and every shift of it), and the other is based on energy measurement taken from a permuted, de-spread received sequence. However, in this paper we develop our system using the preamble-based method.

\subsubsection{Preamble-based Arrival Time Detection} \label{sec:PATD}
 Based on this method, (see Figure~\ref{fig:RandTrans}) we will select the signature sequences to form a set of known sequences (preambles) with a low cross-correlation, at the beginning of each packet (first $MN_p$ bits). A sufficiently long randomly chosen sequence with entries taken from the set $\{-1/\sqrt{M}, 1/\sqrt{M}\}$ can be used as a preamble, which is automatically constructed by multiplying the signature sequence by the repeated-bits sequence, explained in system model section. The preamble in the proposed work is different than the preamble in~\cite{TBKN20} where a compressed sensing method is employed to construct the preamble for the purpose of synchronous URA. After attaching the preamble the packet size changes to $M(N_p+N)$, where the payload length is $MN$. The formation of the payload ($MN$ bits) is  explained in the system model section.

Before starting the interference cancellation process, the receiver first computes the correlation between the received sequence and all possible preamble signals, which are included in the pool, for every symbol shift (or every sample in practice). Therefore, the correlation value for $p$th preamble in the pool, $\vec{P}^{(p)} = [P_1^{(p)}, P_2^{(p)}, \dots, P_{MN_p}^{(p)}]$, is
\begin{align}
	\text{C}_{p,t}^{(j)} = \sum_{t_{\text{det}}=t}^{\text{min}(t+MN_p-1,MN)} y_{t_{\text{det}}}^{(j)}P_{t_{\text{det}}-t+1}^{(p)}+\sum_{t_{\text{det}}=1}^{MN_p-MN+t-1} y_{t_{\text{det}}}^{(j+1)}P_{MN-t+t_{\text{det}}+1}^{(p)}.
\end{align}
where $\text{min}(a,b)=a$ if $a \leq b$, or $b$, otherwise. The timing acquisition correlator detects a packet arrival time every time $\text{C}_{p,t}^{(j)}$ is larger than a pre-selected threshold value ($\text{Thr}_{\text{corr}}$). However, since there is no guarantee that the starting time detection goes perfectly well, a recursive process between the interference cancellation (i.e. multi-user detection of payloads) and starting time detection is necessary. To optimize the system, the recursive process can run until the interference and noise power approaches the noise power $\sigma^2$ which guarantees successful detection of packet arrival times. Afterwards, the interference cancellation will occupy  $100\%$ of the processing power for the remaining decoding iterations. Fig.~\ref{fig:ITD/MUD-CRC} shows the entire receiver diagram and is explained later. The details of the MUD process are given in \cite{TBKN20}.

\subsubsection{Threshold Setting} \label{sec:ThreSet}
A very important factor in detecting packets is the selected threshold value for timing acquisition correlator. Later in Section~\ref{MoreGranAnalysis}, we show that how a threshold, selected for the correlator, controls the probabilities of missed and false packets. Higher thresholds lead to less false packets and at the same time more missed packets. With lower thresholds we detect more real packets, but let the system collect more false packets. This can adversely overwhelm the receiver with a lot of extra interference. In order to optimize the process, we consider an adaptive way of setting the threshold. According to Fig.~\ref{fig:ThreSet} we divide time frame into small segments and set the thresholds for the segments in such a way that the segments with lower interference level have lower threshold values and the segments with higher interference level have higher thresholds. In fact, with lower interference the probability of detecting false packets drops significantly, allowing to set lower threshold and detect a significant number of real packets. Because of coupling way of transmission, iterative interference cancellation leads to the least interference power and the least threshold value for the leftmost of decoding window. The interference increases as we move to the right. For instance, at iteration $i$ for $\tau_2>\tau_1$ the threshold is higher for the segment $\tau_2$: $T_{\tau_2}^{(i)} \geq T_{\tau_1}^{(i)}$. Moreover, this helps a certain segment (for example the segment $\tau$) to have a lower or at most the same threshold value at iteration $i_2$ compared to that at iteration $i_1$, if $i_2 \geq i_1$: $T_{\tau}^{(i_2)} \leq T_{\tau}^{(i_1)}$
\begin{figure}[h]
	\centering
	\setlength{\unitlength}{1mm}
	\begin{picture}(150,115)
	\put(5,5){\includegraphics[scale=0.4]{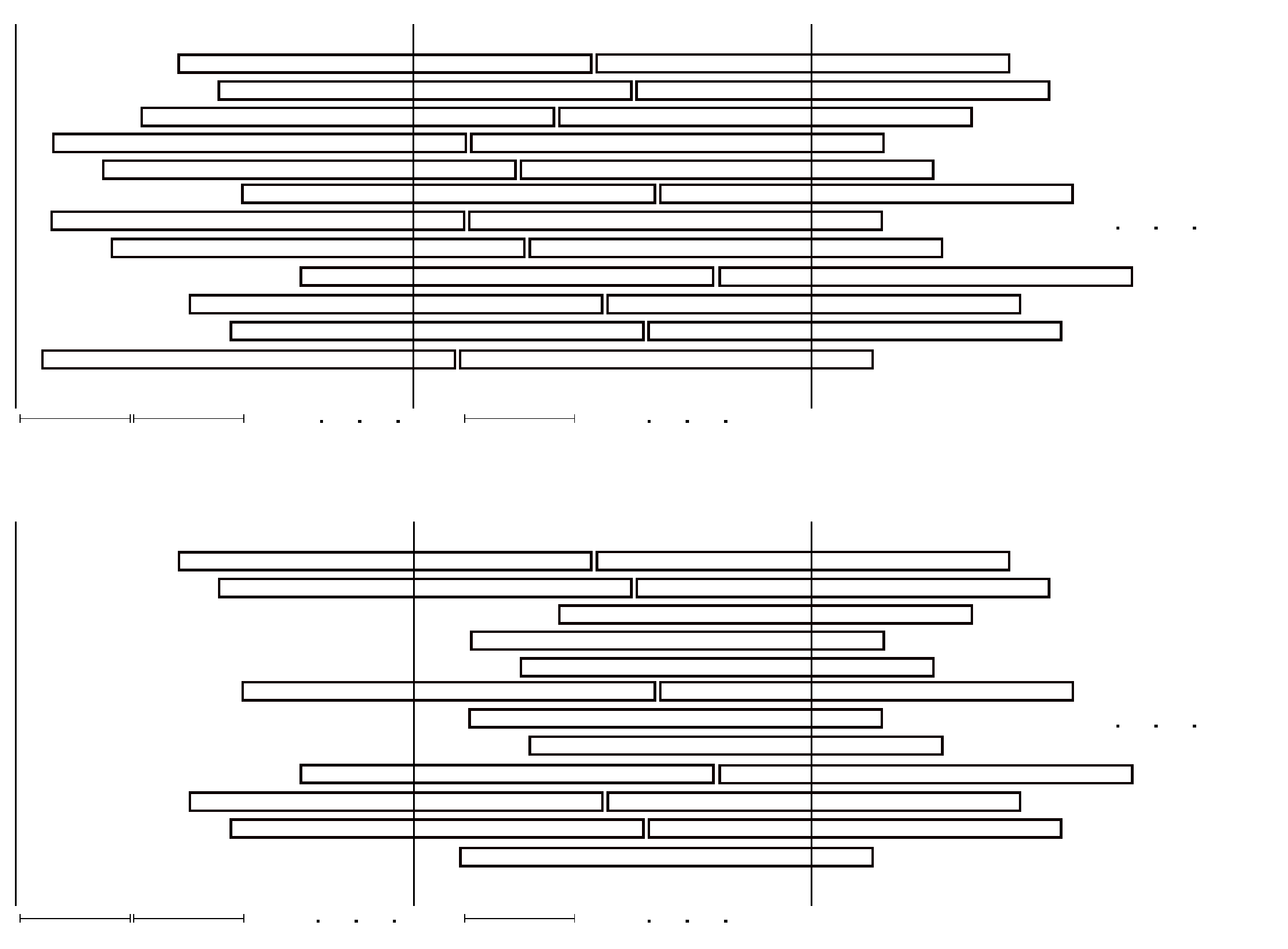}}
	
	\put(-10,117){ Iteration $i_1$}
	\put(9,62){ $T_1^{(i_1)}$}
	\put(17.5,62){ $\leq$}
	\put(22,62){ $T_2^{(i_1)}$}
	\put(30.5,62){ $\leq$}
	\put(45.5,62){\dots}
	\put(57,62){ $\leq$}
	\put(61,62){ $T_{\tau}^{(i_1)}$}
	\put(69,62){ $\leq$}
	
	\put(-10,58){ Iteration $i_2$}
	\put(9,3){ $T_1^{(i_2)}$}
	\put(17.5,3){ $\leq$}
	\put(22,3){ $T_2^{(i_2)}$}
	\put(30.5,3){ $\leq$}
	\put(45,3){\dots}
	\put(57,3){ $\leq$}
	\put(61,3){ $T_{\tau}^{(i_2)}$}
	\put(69,3){ $\leq$}
	
	\end{picture}
	\caption{Adaptive threshold setting: the lower the interference, the lower the threshold value.}
	\label{fig:ThreSet}
\end{figure}
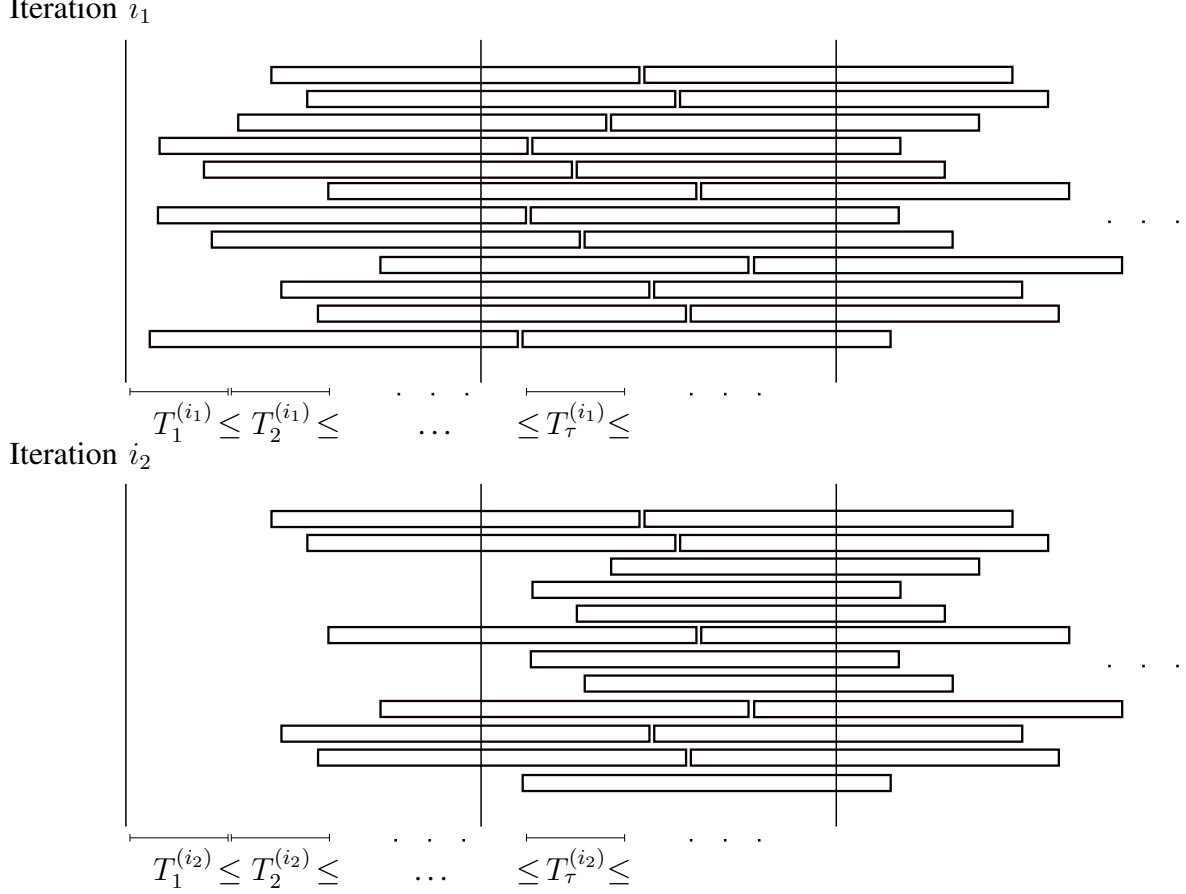

\subsubsection{Iterative Timing Detection/MUD with Cyclic Redundancy Check} \label{sec:ITD/MUD-CRC}
When the number of false packets increases throughout the iterations of timing acquisition and MUD, the resulting interference also increases and leads to performance degradation. In some severe situations with many false packets the receiver can even stop working entirely. Therefore, according to Fig.~\ref{fig:ITD/MUD-CRC}, we consider a modified structure of our system which we call iterative timing detection/MUD with cyclic redundancy check (ITD/MUD-CRC). According to Fig.~\ref{fig:PackStruc} the new approach needs the data bits to be encoded using a CRC code. The CRC bits can help to distinguish between the real and false packets in the receiver. We apply iterative timing acquisition and MUD for several iterations. Then, the packets which pass the CRC check and their related interference are removed from the received sequence. Afterwards, another course of iterative timing acquisition/MUD and CRC step is performed and so on. 
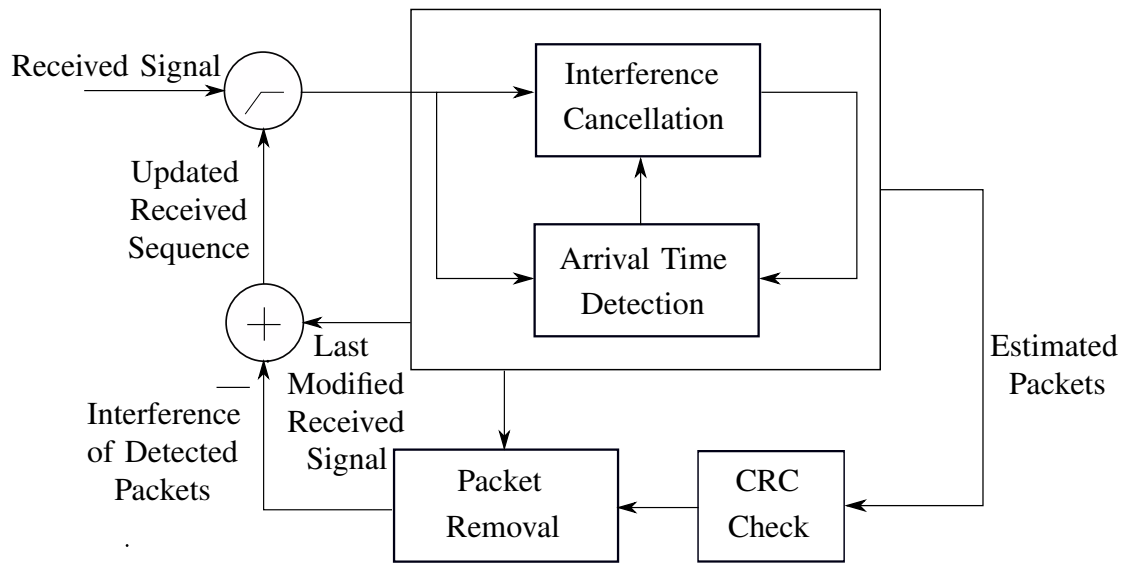
\begin{figure}[h]
	\centering
	\setlength{\unitlength}{1mm}
	\begin{picture}(125,85)
	\put(0,0){\includegraphics[scale=0.35]{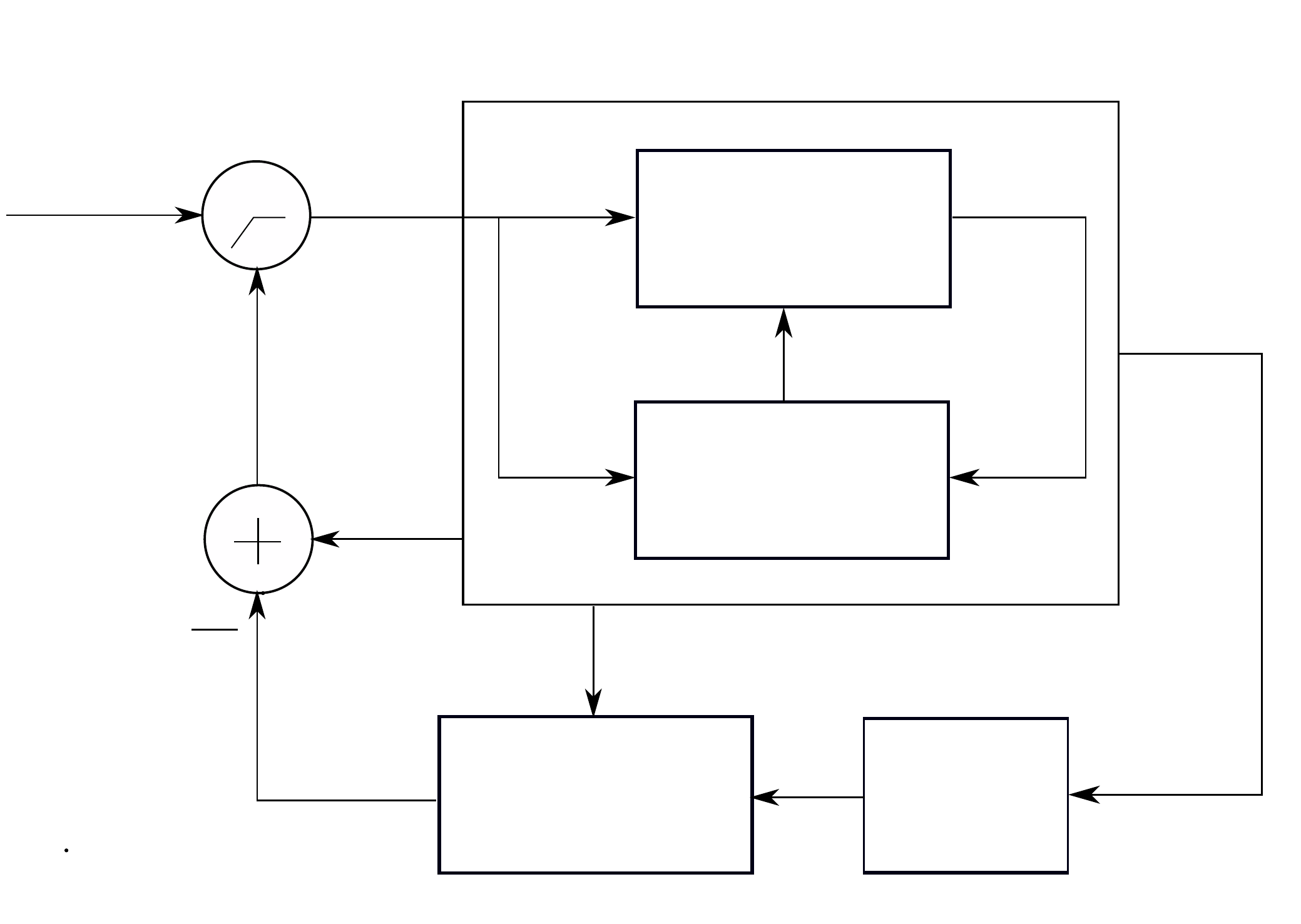}}
	
	\put(64.2,68){Interference}
	\put(64,62){Cancellation}
	\put(63.5,43.5){Arrival Time}
	\put(66.3,37.5){Detection}
	\put(-9,69){Received Signal}
	\put(86.8,14.1){CRC}
	\put(85.8,8.1){Check}
		\put(50,14.1){Packet}
	\put(48.5,8.1){Removal}
		\put(120.5,32){Estimated}
	\put(123,27){Packets}
	
	\put(6.8,55){Updated}
	\put(6.7,50){Received}
	\put(6.4,45){Sequence}
	
	\put(31,32){Last}
		\put(27.5,27){Modified}
	\put(27.5,22){Received}
	\put(30,17){Signal}
	
		\put(1,23){Interference}
	\put(1,18){of Detected}
	\put(4.6,13){Packets}
	\end{picture}
	\caption{Iterative timing detection and multi-user decoding with CRC.}
	\label{fig:ITD/MUD-CRC}
\end{figure}

\begin{figure}[h]
	\centering
	\setlength{\unitlength}{1mm}
	\begin{picture}(125,35)
	\put(16,0){\includegraphics[scale=0.5]{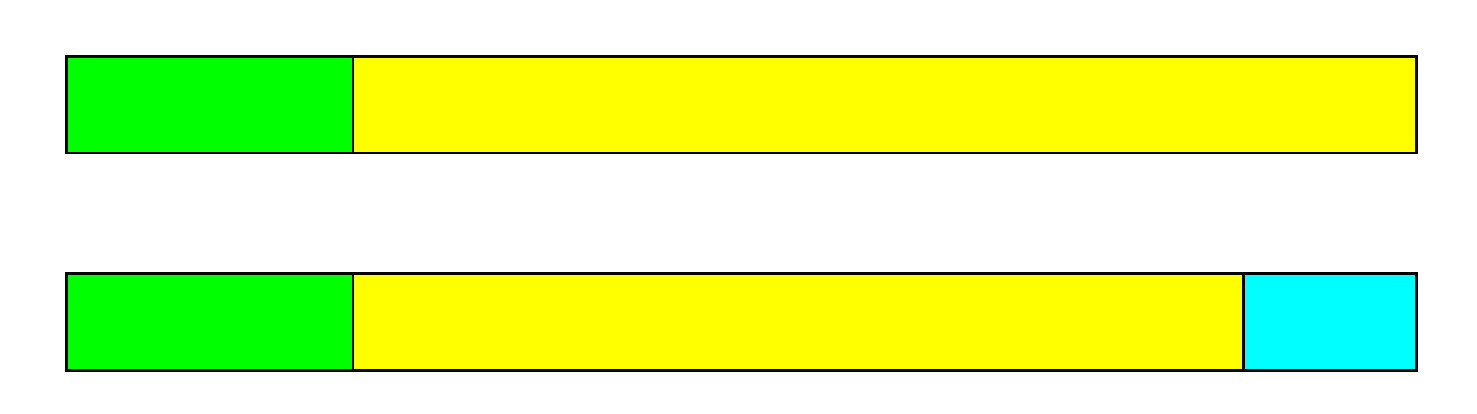}}
	
		\put(25.2,25){Preamble}
	\put(82,25){Payload}
	\put(-18,25){(a) Regular }
	
		\put(25.2,6.5){Preamble}
	\put(76,6.5){Payload}
		\put(124,6.5){CRC}
			\put(-18,6.5){(b) ITD/MUD-CRC }
		\end{picture}
	\caption{Packet structure for (a) the regular and (b) ITD/MUD-CRC cases.}
	\label{fig:PackStruc}
\end{figure}
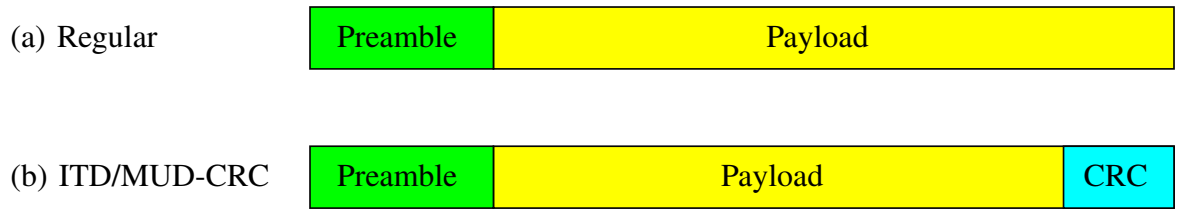

\subsubsection{Collision Reduction} \label{sec:CR}
As we mentioned above, all packets may have payloads modulated with a single pair of permutation and signature sequences in case they arrive at distinct times. This does not lead to any performance degradation. However, in a real transmission scenario there is no control on the times when the active users may transmit, and the probability that at least two packets collide (arrive at the same time) may be high in some parameter settings. In the proposed URA system the same permutation, signature, and arrival time for two packets means collision. Collision has two harmful effects on the performance of the entire system. On one hand, the collided packets will not be decoded successfully. On the other hand, the residual interference resulting from the undetected packets will degrade the performance for the other packets. 

To partially overcome the problem of collisions in the proposed system, we consider $L_p$ pairs of permutation and signature sequences in a pool and allow each user who transmits a packet to choose a permutation and signature pair randomly from the pool before the transmission. As a result, there will be $L_p$ unique preambles as well (see Section~\ref{sec:SM}, which describes packet construction). Therefore, if two packets arrive at the same time there is still chance that their data gets decoded successfully provided that their preambles (and hence permutation-signature pair) are different.

\section{Approach to Analysis of Iterative Receiver Processing}\label{MoreGranAnalysis}
In this section, we outline an approach to the analysis of the iterative receiver. To start, we introduce the following lemma. Let us consider a fixed threshold equal to $85\%$ of $N_p$ (auto-correlation of a preamble bits) for simplicity.

\begin{lemma} Consider $L$ active users of power $1/M$, where the channel is an additive white Gaussian noise (AWGN) channel with noise power $\sigma^2$, and a preamble is of length $MN_p$. The signal to interference and noise ratio (SINR) of the preamble correlator output equals
	\begin{align}
	\text{SNR}_{\text{corr}} = \frac{(N_p)^2}{(L-1)N_p+\sigma^2N_p}=\frac{N_p}{(L-1)+\sigma^2}.\label{eq:SNRCor}
	\end{align}
\end{lemma}
\begin{proof}
	See Appendix~\ref{Ap:Cor}, where we derive formulas for preamble signal's cross-correlation with other signals and the noise, as well as its auto-correlation.
\end{proof}  

As we discussed earlier, it may happen that the preamble detection process will not be able to find all the packets. Also, the detector may detect extra (false) packet(s) by mistake. The probabilities of missing a packet ($p_\text{miss}$) or detecting a false packet ($p_\text{false}$) in the first iteration depend on the $\text{SNR}_{\text{corr}}$ value. Later it will depend on $\text{SNR}_{\text{corr}}^{(i)}$ that will change with the iteration index, $i$. For a large value of $L=K_a$, the number of active users, we can assume, according to the central limit theorem, that the interference and noise is Gaussian with the power equal to $(L-1)N_p+\sigma^2N_p$ and with zero mean, that is $ \mathcal{N}(0,\,(L-1)N_p+\sigma^{2}N_p)$ at the first iteration before the MUD starts. 

False detection happens when there is no packet in the received sequence that starts at a time while the correlation value is larger than $\text{Thr}_{\text{corr}}$.  Therefore, 
\begin{multline}
p_{\text{false}} = \int_{\text{Thr}_{\text{corr}}}^{\infty}\frac{1}{\sqrt{2\pi ((L-1)N_p+\sigma^2N_p)}}\exp \left(\frac{-x^2}{2(L-1)N_p+2\sigma^2N_p} \right) \text{d}x = \\
Q \left( \frac{\text{Thr}_{\text{corr}}}{\sqrt{ (L-1)N_p+\sigma^2N_p}}\right), \label{eq:pfalse_pre}
\end{multline}
where 
\begin{align}
Q(x) = \int_{x}^{\infty}\frac{1}{\sqrt{2\pi}}\exp \left(\frac{-y^2}{2} \right)\text{d}y.
\end{align}
A packet gets missed when it has been transmitted, but the respective correlation value is lower than $\text{Thr}_{\text{corr}}$. Hence, $p_\text{miss}$ can be obtained from 
\begin{multline}
p_{\text{miss}} = \int_{-\infty}^{\text{Thr}_{\text{corr}}}\frac{1}{\sqrt{2\pi ((L-1)N_p+\sigma^2N_p)}}\exp \left(\frac{-(x-N_p)^2}{2(L-1)N_p+2\sigma^2N_p} \right) \text{d}x =\\
1 - Q \left( \frac{\text{Thr}_{\text{corr}}-N_p}{\sqrt{ (L-1)N_p+\sigma^2N_p}}\right). \label{eq:pmiss_pre}
\end{multline}
	
 According to the approach proposed in~\cite{TS19}, if we divide each packet into $2W+1$ sections and all arriving packets, i.e. active users ($L=K_a$), to $2W+1$ different groups, and assume that the groups of packets arrive with equi-distant delays, with $1$ to be the power of each bit ($\frac{1}{M}$ for each replication bit) the interference-and-noise power in position $\tau$ at iteration $i$ equals 
\begin{align}
x_{\tau}^{(i)} = \frac{\alpha}{2W+1} \sum_{\tau_1=-W}^{W} g_{\text{mse}} \left( \frac{1}{2W+1}\sum_{\tau_2=-W}^{W}\frac{1}{x_{\tau+\tau_1+\tau_2}^{(i-1)}}\right)+\sigma^2, \label{eq:xitauo}
\end{align}
where $\alpha=\frac{L}{M}$, $\sigma^2$ is the channel noise power and
\begin{align}
g_{\text{mse}}= \mathbb{E}\left[(1-\text{tanh}(x+\zeta\sqrt{x})\right], \quad \zeta \sim \mathcal{N}(0,1).
\end{align}
is the mean-square error of estimated $+1$ or $-1$ in AWGN of power $\sigma^2$.
We provide an analysis for our proposed system based on Equation~\ref{eq:xitauo}. Two approaches are considered: the regular (iterative timing acquisition and MUD) and the ITD/MUD-CRC approach.

By normalizing the power of each bit to 1 for the regular approach, we find that the interference and noise power in position $\tau$ after interference cancellation at iteration $i$ and before timing acquisition at iteration $i$ equals 
\begin{multline}
x_{\tau}^{(i)} = \frac{1}{2W+1}  \sum_{\tau_1=-W}^{W} \alpha^{(i-1)}_{\tau+\tau_1} g_{\text{mse}} \left(\frac{1}{2W+1}\sum_{\tau_2=-W}^{W}\frac{1}{x_{\tau+\tau_1+\tau_2}^{(i-1)}}\right)\\
+\frac{1}{2W+1} \sum_{\tau_1=-W}^{W}  \frac{\alpha-\alpha^{(i-1)}_{\tau+\tau_1}}{2W+1}\sum_{\tau_2=-W}^{W}\frac{1}{x_{\tau+\tau_1+\tau_2}^{(1)}} + \sum_{i'=1}^{i} y_{\tau}^{(i')} + z_{\tau} + \sigma^2, \label{eq:xitau}
\end{multline}
where $y_{\tau}^{(i')}$ is the interference power from false packets in position $\tau$ at iteration $i'$. We do summation over false packet powers from previous and current iterations because they accumulate from an iterations of timing detection to the other one. $z_{\tau}$ is the collision induced interference power. $\alpha_{\tau}^{(i-1)}$ is the detected load in position $\tau$ at iteration $i-1$. Note that the first term in Equation~\ref{eq:xitau} is related to the interference remaining from the detected packets, and the second term represents the interference of the undetected packets before timing acquisition is done at iteration $i$. All such packets have $x_{\tau'}^{(1)}$ power value (not improved throughout the iterations), since their starting value does not change. For simplicity, we do not distinguish between the newly detected packets and the packets previously detected. 

We can obtain $\alpha^{(i)}_{\tau}$ in
\begin{align}
\alpha^{(i)}_{\tau} = \alpha^{(i-1)}_{\tau} + \left(\alpha-\alpha^{(i-1)}_{\tau}\right)p_{\text{success},\tau}^{(i)}\label{eq:alpha1}
\end{align}
where $p_{\text{success},\tau}^{(i)}$ is the probability that a packet is successfully detected in position $\tau$ at iteration $i$, which can be computed from
\begin{align}
p_{\text{success},\tau}^{(i)} = 1 - p_{\text{miss},\tau}^{(i)}
\end{align}
with $p_{\text{miss},\tau}^{(i)}$ to be the miss probability of packets in position $\tau$ at iteration $i$. To obtain $p_{\text{miss},\tau}^{(i)}$ we refer to Equation~\ref{eq:pmiss_pre} and rewrite it appropriately (for the case where each bit has the power 1) and to include positions and iterations. Therefore,
\begin{align}
p_{\text{miss},\tau}^{(i)} = 1 - Q \left( \frac{\text{Thr}_{\text{corr}}^{(i)}-N_p}{\sqrt{N_p x_{\tau}^{(i)}}}\right), \label{eq:pmiss_pre_positer}
\end{align}

Unfortunately,~(\ref{eq:alpha1}) does not appropriately estimate the value of $\alpha_{\tau}^{(i)}$. There is always an amount to be added to $\alpha_{\tau}^{(i-1)}$ since $p_{\text{success},\tau}^{(i)}$ is not 0 for $x_{\tau}^{(i)}<\infty$. Therefore, as an alternative to~(\ref{eq:alpha1}) we use 
\begin{align}
\alpha^{(i)}_{\tau} = \alpha p_{\text{success},\tau}^{(i)}\label{eq:alpha2}
\end{align}

We compute $z_{\tau}$ in
\begin{align}
z_{\tau} = \frac{1}{M}\sum_{\tau'=\tau-2W}^{\tau}N_{\textrm{col},\tau'}, 
\label{eq:ztau}
\end{align}
where $N_{\textrm{col},\tau'}$ is the average number of collided packets in position $\tau'$, obtained from
\begin{align}
N_{\textrm{col},\tau'} = 2p_{\textrm{col:2},\tau'}+3p_{\textrm{col:3},\tau'}+...+\alpha_{\tau'}^{(1)}Mp_{\textrm{col:}\alpha_{\tau'}^{(1)}M,\tau'}.
\label{eq:Ncol}
\end{align}
In~(\ref{eq:Ncol}), $p_{\textrm{col:j},\tau'}$, with $j=1,2,\dots,\alpha_{\tau'}^{(1)}M$, is the probability that only $j$ packets collide in position $\tau$, while no other packets collide in that position. It can be computed experimentally using a computer simulation, or using the following Lemma.

\begin{lemma} Consider $L_{\tau}$ packets which can start at any time instant within the discrete interval $\tau$ with the length $n_\tau$. The probability that $q$ of these packets start concurrently (at the same time instant) in $\tau$ is:
\begin{align}
 p_{\textrm{col:q},\tau} = \frac{(L_{\tau}-q+1)!{\binom{n}{L_{\tau}-q+1}}\binom{L_{\tau}}{q}}{n_{\tau}^{L_{\tau}}}
\label{eq:p_col}
\end{align}
\end{lemma}

\begin{proof}
	See Appendix~\ref{Ap:p_col}.
\end{proof}  

We have only considered two packet collisions because of their dominant probabilities compared to the other cases (three packet collisions, four packet collisions, etc.). We have computed and for different $L=K_a$ values, which are shown in Table. Note that by selecting $W=5$, we divide each packet into 11 segments. Therefore, $L_{\tau}=\nint{\frac{L}{11}}$, $n_{\tau}=\nint{\frac{B_{\text{p}}}{11}}$, where $\nint{}$ returns the nearest integer value of its argument. For, $B_{\text{p}}=10000$, $n_{\tau}=909$, and for $B_{\text{p}}=30000$, $n_{\tau}=2727$. Tables~\ref{tab:Ncol_10k} and~~\ref{tab:Ncol_30k}, respectively, provide the values of $p_{\textrm{col:2}}$ and $N_{\textrm{col},\tau}$ for these two cases.

\begin{table}[!h]
\tiny
\centering
\caption{Average number of collided packets in  position $\tau$ for $B_{\text{p}}=10000$}
\begin{tabular}{| c | c c c c c c c c c c c c c|}
 \hline
 $L$ & 30 & 50 & 70 & 80 & 90 & 100 & 110 & 125 & 140  & 150 & 165 & 170 & 180 \\ 
 \hline
 $L_{\tau}$ & 3 & 5 & 6 & 7 & 8 & 9 & 10 & 11 & 13 & 14 & 15 & 15 & 16 \\
 \hline
 $p_{\textrm{col:2}}$ & 0.003 & 0.0109 & 0.016 & 0.0227 & 0.03 & 0.0384 & 0.0476 & 0.0576 & 0.0798 & 0.0918 & 0.1045 & 0.1045 & 0.1175\\
 \hline
 $N_{\textrm{col},\tau}$ & 0.006 & 0.0218 & 0.032 & 0.0454 & 0.06 & 0.0768 & 0.0952 & 0.1152 & 0.1595 & 0.1837  & 0.2089 & 0.2089 & 0.2351\\
 \hline
\end{tabular} \label{tab:Ncol_10k}
\end{table} 

\begin{table}[!h]
\tiny
\centering
\caption{Average number of collided packets in  position $\tau$ for $B_{\text{p}}=30000$}
\begin{tabular}{| c | c c c c c c c c c c c c c c c |}
 \hline
 $L$ & 30 & 50 & 70 & 90 & 110 & 125 & 140  & 165 & 170  & 200 & 280 & 400 & 420 & 440 & 500\\
 \hline
 $L_{\tau}$ & 3 & 5 & 6  & 8  & 10 & 11 & 13 & 15 & 15 & 18 & 25 & 36 & 38 & 40 & 45\\
 \hline
  $p_{\textrm{col:2}}$ & 0.0011 & 0.0037 & 0.0055  & 0.0102  & 0.0163 & 0.0198 & 0.0279 & 0.0372 & 0.0372 & 0.0533 & 0.0994 & 0.1855 & 0.2017 & 0.2177 & 0.2561 \\
 \hline
 $N_{\textrm{col},\tau}$ & 0.0022 & 0.0073 & 0.011  & 0.0204  & 0.0326 & 0.0397 & 0.0558  & 0.0745 & 0.0745 & 0.1067 & 0.1988 & 0.3711 & 0.4034 & 0.4354 & 0.5122 \\
 \hline
\end{tabular} \label{tab:Ncol_30k}
\end{table} 

Note that collisions introduce an additional fixed interference power term to the total interference and noise power throughout the iterations. On the other hand, the packets which are included in the collision term must be excluded from the rest of the transmitted packets~(\ref{eq:xitau}). Therefore, we adjust the value of $\alpha$ in accordance with
\begin{align}
\alpha = \alpha -\frac{N_{\textrm{col},\tau}}{M}.
\label{eq:alpha_mod}
\end{align}

\section{Numerical Results}\label{sec:numericalResults}
For different preamble lengths we have computed the maximum number of active users, $L=K_a$, for which the timing acquisition and interference cancellation units work successfully together to decode the data of the users. The per-user probability of error (PUPE) is set to be approximately 0.05 or less. For this experiment the length $M(N+N_p)$ of transmitted packet size is approximately $10000$ with $N=100$ information bits. Hence, for packet sizes of $20, 25, 30, 35, 40, 45, 50 ,55$ and $60$ the equivalent repetition factors, $M$'s, are $83, 80, 77, 74, 71, 69,67, 65$ and $63$, respectively. The collision events (two packets arriving at the same time) are not included in this experiment, since we aim to find the best preamble length (and we assume it should be the best also for the case with collisions). Additionally, we have considered an adaptive threshold setting for the timing acquisition block; at iteration $i$, $i=1,2,\cdots$ we compute the average noise and interference power for every block within the window, and set the threshold for each such that the lower noise and interference powers lead to lower threshold values and vice versa. Our numerical results show that preamble length $N_p=30$ is a good choice which results in the higher overall numbers of supported users.  since it provides more users to transmit concurrently. 

Fig.~\ref{fig:SNRvsL} shows the results in terms of the lowest SNR ($\textrm{E}_b/\textrm{N}_0$) for which a certain number of active users ($L=K_a$) can concurrently transmit their packets through the channel. The results are plotted as functions of the SNR with respect to $L$. The red curve corresponds to the controlled delays system, where the starting times of the packets are known. We note that, sometimes depending on the channel situation and the number of transmitted packets, the timing detection and interference cancellation totally fail to function. In this situation there is a need for the users and the receiver to restart transmitting and receiving packets. The number of erroneous blocks resulting from this case which contributes to the total PUPE is $W_{\text{in}}K_a$, where $W_{\text{in}}$ (typically $W_{\text{in}}=5$) is the MUD and decoding window length.
The cyan curve with squares, the black and blue curves are related to the ITD/MUD-CRC discussed in the previous section with the restarting process, while they respectively show the results for the cases when the pools of 1, 4 and 8 preamble(s) are utilized. The utilized CRC part is of 8 bits generated using $X^8+X^7+X^6+X^4+X^2+1$ polynomial. Comparing these 3 curves shows that the pool of 4 preambles outperforms the other two pool choices. The reason why it is better than the pool of 8 preambles may be the higher false alarm probability of the pool 8 case. The ITD/MUD-CRC curves all follow the red curve (known timing) with some distance until at some point the curves bend over. The reason for this can be the smaller repetition factor $M$ which needs to be used for the CRC-based system, compared to that of the system with known timing. Because of the attached CRC bits in order to have the same packet size (10000 bits) the value of $M$ have to be smaller for all the cases which use CRC.
The yellow and green curves reflect the results for known delays system with coding and ITD/MUD-CRC with coding and restarting strategy while employing a pool of 4 preambles. The utilized code for the former is a shortened Hamming code (109,100), and for the latter is shortened Hamming code (117,108), where the first number points out to the code length and the second number is the information length. The orange curve with triangles is related to the regular transmission where we do not employ CRC bits.

The system with coding has the best performance for the low SNR values. However, since the total length of the coded bits is longer for the packet size to satisfy the 10000 length limitation the repetition factor, $M$, must decrease even more. This can cause performance degradation for higher SNR values for the coded system.

\begin{figure}[h]
	\centering
	\setlength{\unitlength}{1mm}
	\begin{picture}(120,90)
	\put(-3,0){\includegraphics[scale=0.8]{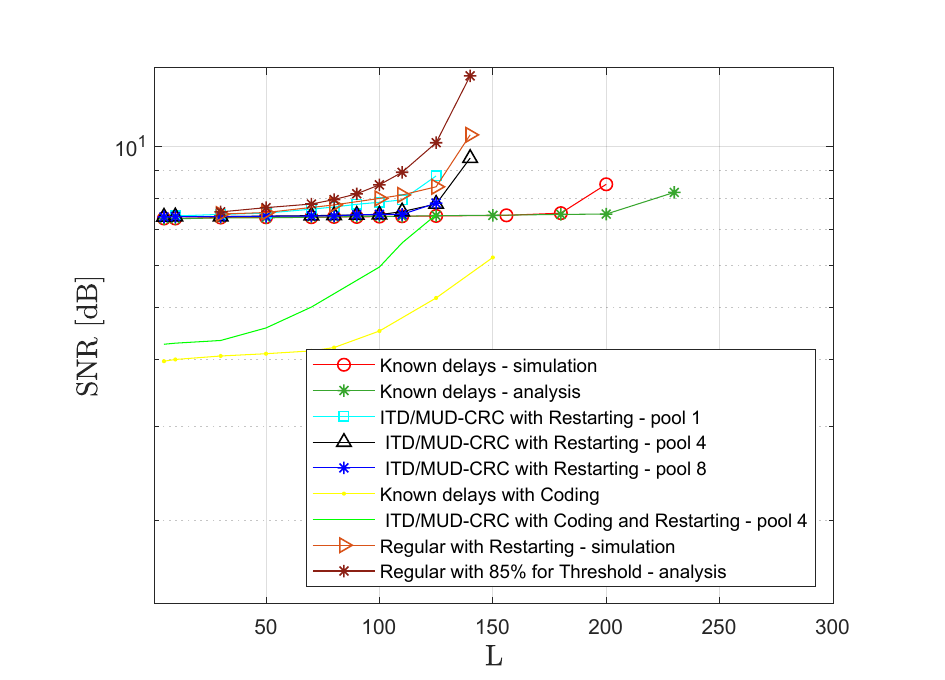}}
	
	\end{picture}
	\caption{SNR ($\textrm{E}_b/\textrm{N}_0$) versus highest number of users, $L=K_a$, for which  0.05 of PUPE (successful decoding convergence) is guaranteed.}
	\label{fig:SNRvsL_yInves}
\end{figure}
The cyan curve with squares and the dotted curve reflect the results for the analysis by neglecting collision-related term in~(\ref{eq:xitau}). The cyan is related to the case where we use~(\ref{eq:alpha1}) to update the value of $\alpha^{(i)}_{\tau}$, but the other employs~(\ref{eq:alpha2}) to do so.

To investigate the effect of false packets in the total performance of the system we provide variety of curves for the pools with 1,2 and 4 preambles in Fig.~\ref{fig:SNRvsL_yInves}. We have assumed no packet collisions for all of these results. The red dashed curves show the analysis results with the effect of false packets included. The blue curves with circles relate to the case where false packets are included and all real packets are assumed to be detected even before the first iteration of interference cancellation is performed. The brown curves with stars have the similar properties with the latter case, but no false packet detection is assumed. The green curves with squares represent the collision-free simulation results. Comparing the analysis curves, the adverse effect of false packets in performance is clearly observed especially for higher number of users, which is equivalent to larger interference. 

\begin{figure}[h]
	\centering
	\setlength{\unitlength}{1mm}
	\begin{picture}(120,50)
	\put(-25,5){\includegraphics[scale=0.4]{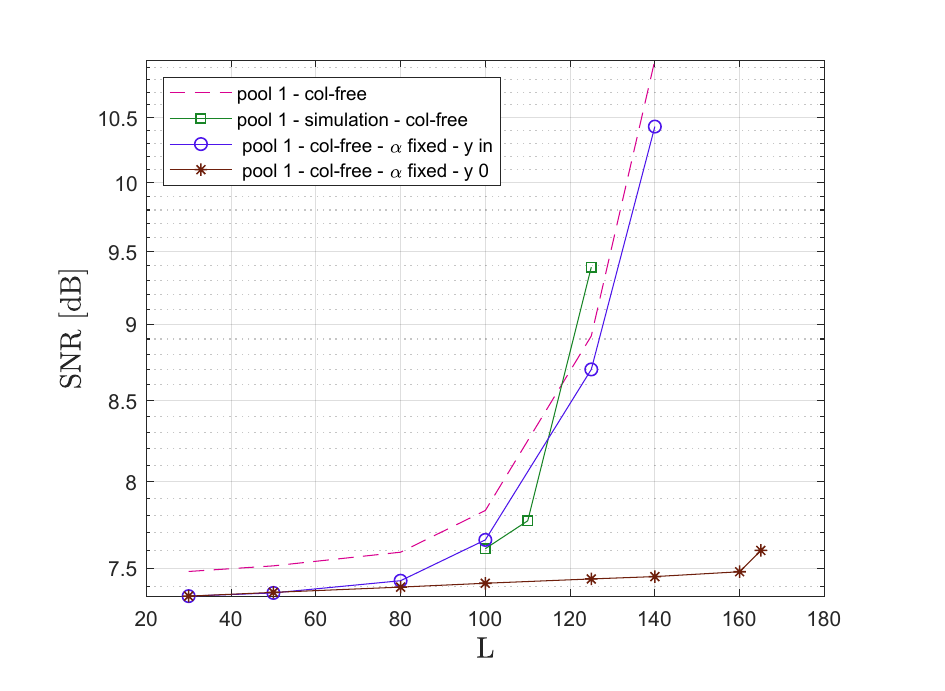}}
	
	\put(29,5){\includegraphics[scale=0.4]{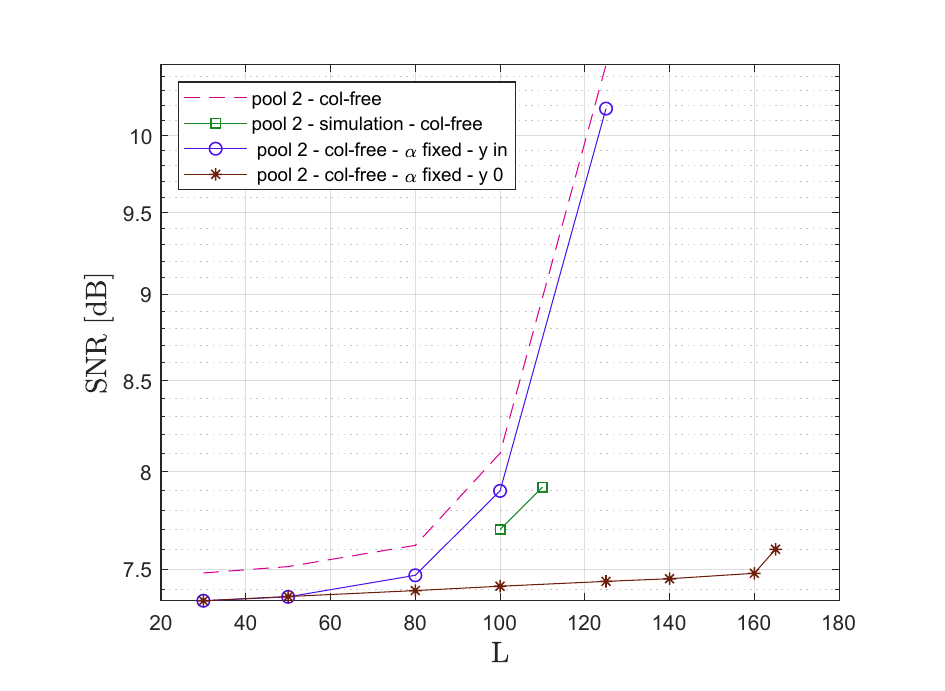}}
	
	\put(85,5){\includegraphics[scale=0.4]{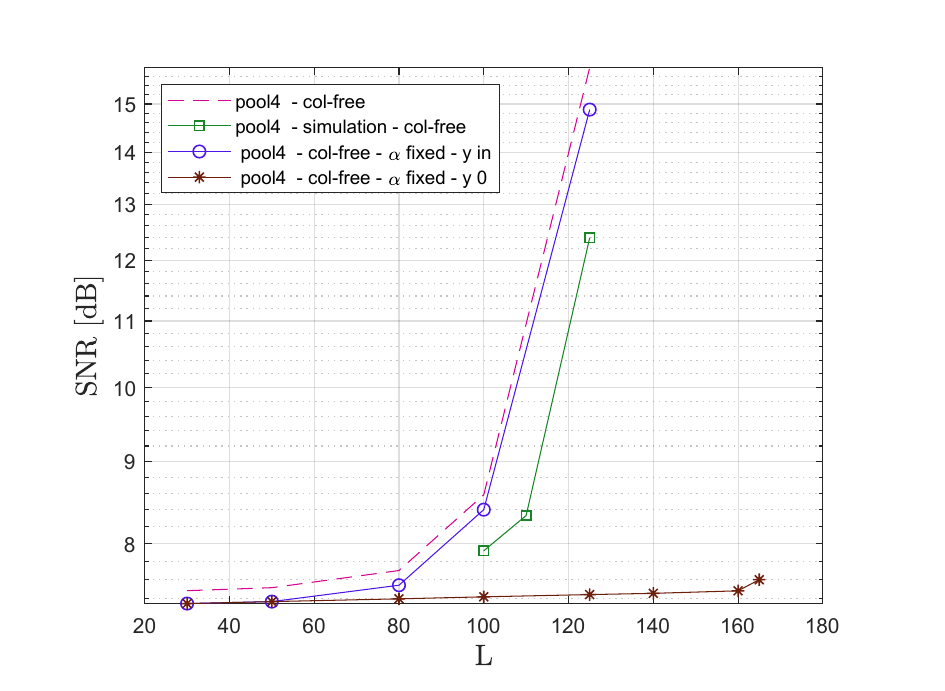}}
	
	\put(5.55,2){(a)}
	\put(60.55,2){(b)}
	\put(115.45,2){(c)}
	
	\end{picture}
	\caption{SNR versus $L$ for (a) pool with 1 preamble, (b) pool with 2 preamble and (a) pool with 4 preamble}
	\label{fig:SNRvsL}
\end{figure}
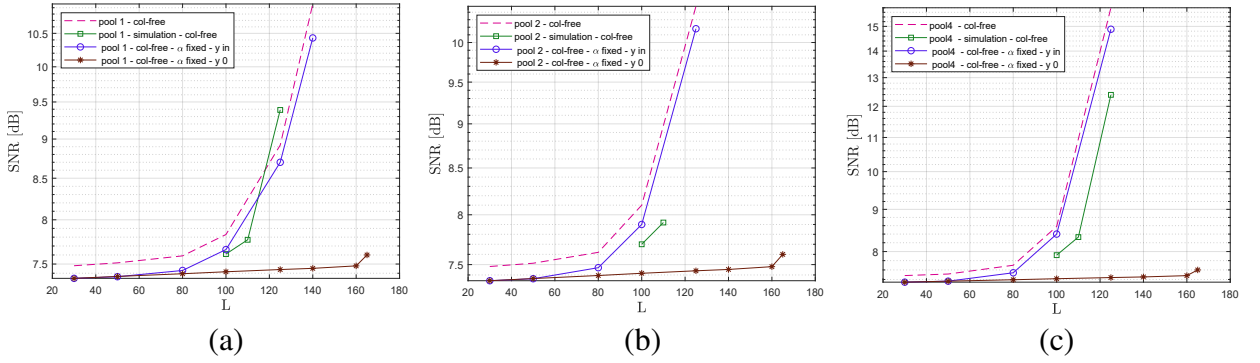

\section{Conclusion}
\label{sec:Conclusion}
In this paper, we proposed a fully asynchronous URA transmission system and test and analyze it for the AWGN channel. The proposed packet format consists of a preamble used for correlation-based timing acquisition and a payload modulated via bit repetition, permutation and scrambling. A small pool of preambles which correspond to distinct choices of permutation and scrambling sequences of the payload is proposed for collision resolutions. CRC bits are applied on the payload to facilitate the interference cancellation at the receiver.  Iterative timing acquisition and multiple user detection receiver is constructed and studied. The receiver operates in a sliding window fashion and takes advantage of natural spatial coupling of the payload graphs. For a wide range of active user loads the proposed receiver is capable to deliver a performance close to that of the system with know packet timing.


 \begin{appendices}
	
\section{Proof of Lemma 1}\label{Ap:Cor}

Suppose $\vec{a}_i=(a_{i1}, a_{i2}, \dots, a_{iMN_p})$, with $i=1,\cdots,L$, are random sequences of length $MN_p$ with the entries chosen from set $\{-1/\sqrt{M}, 1/\sqrt{M}\}$. The variance of the cross-correlation between $\vec{a}_j$ and $\sum_{i=1,i \neq j}^{L} (\vec{a}_i)$ is 
\begin{align}
	\text{var}\left(\sum_{i=1,i \neq j}^{L} \vec{a}_j.\vec{a}_i\right) = \mathbb{E}\left(\sum_{i=1,i \neq j}^{L} \vec{a}_j.\vec{a}_i\right)^2 - \left(\mathbb{E}\sum_{i=1,i \neq j}^{L} \vec{a}_j.\vec{a}_i\right)^2. \label{eq:var_cross}
	\end{align}
	Since $\mathbb{E}(\vec{a_j}.\vec{a_i})=0$ for $i \neq j$, 
	\begin{align}
	\text{var}\left(\sum_{i=1,i \neq j}^{L} \vec{a}_j.\vec{a}_i\right) = \mathbb{E}\left(\sum_{i=1,i \neq j}^{L}
	(a_{j1}a_{i1}+a_{j2}a_{i2}+\dots+a_{jMN_p}a_{iMN_p}) \right)^2 
	\end{align}
	For any $1\leq k_1,k_2,k_3,k_4 \leq L$ where at least two of them are not equal, $\mathbb{E}(a_{k_1}a_{k_2}a_{k_3}a_{k_4})=0$. Therefore,
	\begin{align}
	\text{var}\left(\sum_{i=1,i \neq j}^{L} \vec{a}_j.\vec{a}_i\right) =
	\mathbb{E}\left(\sum_{i=1,i \neq j}^{L}
	(a_{j1}^2a_{i1}^2+a_{j2}^2a_{i2}^2+\dots+a_{jMN_p}^2a_{iMN_p}^2) \right).
	\end{align}
	
	\begin{align}
	\text{var}\left(\sum_{i=1,i \neq j}^{L} \vec{a}_j.\vec{a}_i\right) = (L-1)N_p.
	\end{align}
	Similarly the variance of the correlation of $\vec{a}_i$ with the AWGN is
	\begin{align}
	\text{var}(\vec{a}_j.\vec{n}) = \mathbb{E} (a_{j1}n_1+a_{j2}n_2+\dots+a_{jMN_p}n_MN_p)^2=N_p\sigma^2
	\end{align}
	
	where $n_k$, $1 \leq k \leq MN_p$, are the noise vector entries. And simply the auto-correlation of $\vec{a_i}$ is $N_p$ with the power of $N_p^2$.

\section{Proof of Lemma 2}\label{Ap:p_col}

According to Fig. there are $L_{\tau}$ lines (equivalent to packet positions) with $n_\tau$ points (equivalent to possible starting times) on them. We first compute the number of cases that only the first points from all the first $q$ lines are selected (only $q$ collisions in the beginning of packet positions). One example of such cases is shown in Fig. by the red circles around the points. The total number of them is \begin{align}
 N'_{1} = \left(\n_{\tau}-(L_{\tau}-q)\right)\left(\n_{\tau}-(L_{\tau}-q)+1\right)\dots(\n_{\tau}-2)(\n_{\tau}-1)
\end{align}
Now, suppose all the second points from the first $q$ lines are selected  (only $q$ collisions in the second instance of packet positions). It is clear that the number of such cases, $N'_2$, is equal to $N'_1$, because there is no preference between the first and second instances. The same is true for the other instances ($N'_{n_\tau} = \dots = N'_3 = N'_1$). Therefore, the total number of $q$ collisions in the first $q$ lines is 
\begin{align}
 N'_{T,1} &= \sum_{}^{} \left(\n_{\tau}-(L_{\tau}-q)\right)\left(\n_{\tau}-(L_{\tau}-q)+1\right)\dots(\n_{\tau}-2)(\n_{\tau}-1)n_{\tau} \\
& =  (L_{\tau}-q+1)!\binom{n_{\tau}}{L_{\tau}-q+1} 
\end{align}
By allowing to select $q$ collided points from different lines, the total number of $q$ collision cases is
\begin{align}
 N'_{T} = (L_{\tau}-q+1)!\binom{n_{\tau}}{L_{\tau}-q+1} \binom{L_{\tau}}{q}.
\end{align}
Finally, the probability of $q$ collisions is obtained from
\begin{align}
p_{\textrm{col:q},\tau} = \frac{(L_{\tau}-q+1)!{\binom{n}{L_{\tau}-q+1}}\binom{L_{\tau}}{q}}{n_{\tau}^{L_{\tau}}}.
\end{align}
where the denominator is the total number of cases of choosing one point from each of $L_{\tau}$ lines (in total $L_{\tau}$ points).
\begin{figure}[h]
	\centering
	\setlength{\unitlength}{1mm}
	\begin{picture}(100,75)
	\put(0,5){\includegraphics[scale=0.8]{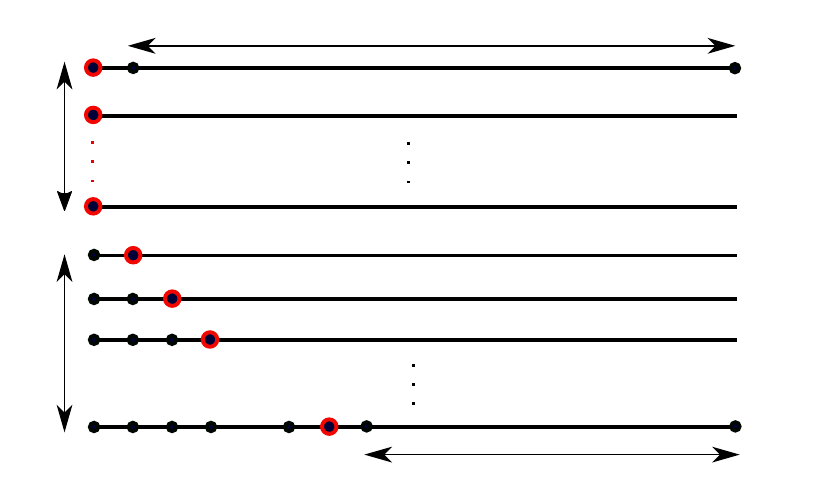}}
	\put(4.5,53){$q$}
	\put(-5,25){$L_{\tau}-q$}
	\put(50,69){$n_{\tau}-1$}
	\put(58,6){$n_{\tau}-(L_{\tau}-q+1)$}
	
	\put(103,62.7){$1$}
	\put(103,55.7){$2$}
	\put(103,44){$q$}
	\put(103,38){$q+1$}
	\put(103,32){$q+2$}
	\put(103,26.2){$q+3$}
	\put(103,14){$L_{\tau}$}
	\end{picture}
	\caption{}
	\label{fig:SNRvsL}
\end{figure}
 \end{appendices}

\bibliographystyle{IEEEtranNodash}
\bibliography{CoupBlock}


\stepcounter{page}
\label{lastpage}
\addtocounter{page}{-1}

\end{document}